\documentclass[a4paper, USenglish, cleveref, autoref, thm-restate]{lipics-v2021}

\title{Randomized Separations in Black-Box TFNP}

\author{Fedor {Kiselev}}{MIPT, Moscow, Russia}{kiselev.fa@phystech.edu}{https://orcid.org/0009-0008-0515-2456}{}

\authorrunning{F. Kiselev} 

\Copyright{Fedor Kiselev} 



\nolinenumbers 

\hideLIPIcs
\usepackage[classfont=sanserif,langfont=caps,funcfont=roman]{complexity}

\usepackage{csquotes}
\usepackage{mathtools}

\let\emptyset\varnothing

\newcommand{\eps}{\varepsilon}

\let\temp\phi
\let\phi\varphi
\let\varphi\temp

\newcommand{\restrharpoon}{\mathord{\upharpoonright}}

\newcommand{\smallupuparrows}{\raisebox{0.0ex}{\scalebox{1}[0.7]{$\upuparrows$}}}

\newcommand{\agree}[1]{\overset{#1}{\smallupuparrows}}

\DeclarePairedDelimiter\floor{\lfloor}{\rfloor}

\DeclarePairedDelimiter\autobracket{(}{)}
\newcommand{\br}[1]{\autobracket*{#1}}

\newclass{\TFAP}{TFAP}
\newclass{\PWPP}{PWPP}
\newlang{\Pigeon}{Pigeon}
\newlang{\InjPigeon}{InjPigeon}
\newlang{\BijPigeon}{BijPigeon}
\newlang{\Lonely}{Lonely}
\newlang{\Lossy}{Lossy}

\usepackage{tikz}
\usepackage{caption}
\usepackage{subcaption}
\usetikzlibrary{positioning}
\usetikzlibrary{arrows.meta}
\definecolor{processblue}{cmyk}{0.96,0,0,0}
\usetikzlibrary{matrix}
\usetikzlibrary{decorations.pathreplacing,angles,quotes}
\usetikzlibrary{calc}

\tikzset{%
	>={Latex[width=1mm,length=1mm]},
}
\tikzstyle{circ} = [draw, circle, fill=gray!15, minimum size=3mm]
\tikzstyle{square} = [draw, rectangle, fill=gray!15, minimum size=3mm]
\tikzstyle{arrow} = [-{Stealth}]
\tikzstyle{bi-arrow} = [{Stealth}-{Stealth}]
\def\shift1{3mm}

\begin{document}

\maketitle

\begin{abstract}
We study the relationship between deterministic and randomized black-box reducibility between problems in $\TFNP$. Our main contribution is a general technique that establishes equivalence between these reducibility types from specific $\TFNP$ problems to any $\TFNP$ problem. In particular, we show that this equivalence holds for reductions from complete problems in $\PPP$, $\PPAD$, $\PPA$, and $t$-$\PPP$. In turn, it strengthens all known black-box separations, originating from these classes, to randomized separations.
\end{abstract}

\section{Introduction}
$\TFNP$ is the class of search problems such that for any input there is at least one correct solution, and the correctness of solutions can be checked in polynomial time. It contains many important problems, such as integer factorization and finding a Nash equilibrium. Being a semantic class, $\TFNP$ is conjectured to have no complete problems \cite{Pap94}. Consequently, most of the research is focused on studying subclasses of $\TFNP$, defined by complete problems in them. The most well-studied of these classes are $\PLS$ \cite{JPY88}, $\PPP$, $\PPA$, $\PPAD$, and $\PPADS$ \cite{Pap94}. 

It is straightforward to show that $\PPAD \subset \PPADS \subset \PPP$ and $\PPAD \subset \PPA$ \cite{Pap94}, and it is conjectured that there are no other inclusions between them. However, proving such separations unconditionally would imply $\P \neq \NP$. Instead, we can look at black-box analogs of these classes. Since all known inclusions between $\TFNP$ classes also hold for their black-box analogs, separations between black-box classes rule out known techniques for showing inclusion.

After years of progress, for every pair of the above mentioned classes, it has been shown that either the inclusion holds unconditionally or it does not hold in the black-box setting \cite{BCE98, Mor01, Mor04, GHJ24}. Subsequently, this result was extended to all possible intersections of these classes \cite{Feasible_Disjunction, Intersection_Classes}. Naturally, a question arises whether these black-box separations could be further strengthened. One possible way to do this is to show that these separations persist even if we consider randomized reductions between complete problems in classes instead of deterministic ones. Such stronger separations were achieved for some pairs of classes: in \cite{Li24} it was shown that in the black-box setting there are no randomized reductions from any of the above mentioned $\TFNP$ subclasses to $\PWPP$ (introduced in \cite{JERABEK2016380}), and in \cite{Jain2024OnPP} it was shown that there is no randomized reduction from $\PPP$ to the newly introduced class $n$-$\PWPP$. In addition, for most of the studied $\TFNP$ classes, it is easy to show that there is no randomized reduction from them to $\FP$. A notable exception is the problem $\Lossy$, introduced in \cite{Korten_Derandomization}, which admits a randomized reduction to $\FP$ but no deterministic one. It demonstrates that these two relations do not always coincide. Overall, randomized reducibility in the black-box $\TFNP$ remains largely unexplored.

\subsection{Our Results}

We investigate the conditions under which it is possible in the black-box model to strengthen a deterministic separation to a randomized one. Our main result, stated in \autoref{theorem:main}, shows that for certain well-structured problems, if there is no reduction to some $S \in \TFNP$, there is also no randomized reduction to $S$. Notably, $S$ is an arbitrary problem in $\TFNP$ with no additional restrictions. The proof technique applies to highly symmetric problems. In addition, it requires that, under particular restrictions, the problem can be reduced to a smaller instance of the same problem. Many central $\TFNP$ problems possess these required structural properties.

We then demonstrate that all the requirements of \autoref{theorem:main} are met for certain complete problems in classes $\PPP, \PPA, \PPAD$ \cite{BCE98}, and $t$-$\PPP$ \cite{Jain2024OnPP}. In particular, it automatically strengthens all known separations originating from them to randomized separations.

\subsection{Connection to Proof Complexity}

Black-box $\TFNP$ has a very tight connection with proof complexity. For many $\TFNP$ subclasses it has been shown that they are characterized by some proof system \cite{GHJ24}. That is, problem $S \in \TFNP$ is in the subclass if and only if the totality of $S$ has an easy proof in the proof system corresponding to that subclass. Moreover, in \cite{TSPinZPP} it was shown that this connection can be extended to randomized reductions, using randomized versions of proof systems. More specifically, problem $S \in \TFNP$ is randomized-reducible to a complete problem in some subclass if and only if the totality of $S$ has an easy proof in the randomized version of the corresponding proof system. This connection can be used to translate our results to proof complexity. Let us call a proof system well-behaved if it has corresponding proof system (see more on that in \cite{TFNPchar}). Then for any problem $S$ such that it meets the conditions of \autoref{theorem:main}, if the totality of $S$ is hard to prove in some well-behaved proof system, then it is also hard to prove in its randomized version.

Finally, it should be mentioned that the result, similar to ours, was achieved in \cite{RandomResRef}, where it was essentially shown that $\PPP$ is not randomized-reducible to $\PLS$.

\subsection{Paper Organization}
The remainder of the paper is structured as follows. In Section 2 we introduce most of the required definitions. In Section 3 we prove our main theorem, as well as a few auxiliary results. In Section 4 we define classes $\PPP, \PPAD, \PPA$ and $t$-$\PPP$ and demonstrate that the main theorem is applicable to them. In Section 5 we briefly discuss limitations of the main theorem and future directions.

\section{Preliminaries}
\subsection{Black-Box TFNP}
Black-box classes are usually distinguished from their white-box analogs by using superscript \enquote{dt}, which stands for \enquote{decision tree}. Since our work is dedicated to black-box classes only, we will omit this superscript.

\begin{definition}
    A \textit{search problem} is a sequence $R = \{R_N\}_{N\in \mathbb{N}}$, where $R_N \subset \Sigma_N^{l_N} \times O_N$ for some finite \textit{alphabet} $\Sigma_N$, a finite \textit{set of solutions} $O_N$, and \textit{input length} $l_N$. Any $x \in \Sigma_N^{l_N}$ is called the \textit{input} of $R_N$. A solution $o$ is \textit{correct} for an input $x$ if $(x,o)\in R_N$. A search problem is considered \textit{total} if, for any input, there is a correct solution. Search problem $R$ belongs to $\TFNP$ if it is total and there is a family of decision trees $T = \{T_{N, o}\}_{N\in \mathbb{N}, o\in O_N}$ of depth $\poly(\log N)$ such that $(x,o)\in R_N \Leftrightarrow T_{N,o}(x) = 1$.
\end{definition}

Note that this definition is more general than the commonly used one, where $\Sigma_N = \{0,1\}$ and $l_N = N$. We believe that this generalization makes reasoning easier since it allows us to omit the binary encoding of input. All problems defined in \autoref{section:application} can be transformed into equivalent problems with $\Sigma_N = \{0,1\}$ and $l_N = N$.

We will sometimes abuse the notation by calling $R_N$ a search problem. Also, if there are several search problems in the context, we will distinguish related objects by subscript letters.

\begin{definition}
    Given search problems $R_N$ and $S_M$, a \textit{pseudo-reduction} from $R_N$ to $S_M$ is a pair $(f = \{f_i\}_{i\in [l_M]}, g = \{g_o\}_{o\in O_M})$ of families of decision trees over inputs of $R_N$. Each $f_i$ defines a mapping $\Sigma_N^{l_N}\to \Sigma_M$; each $g_o$ defines a mapping $\Sigma_N^{l_N}\to O_N$. We say that the pseudo-reduction is successful on input $x$ of $R_N$ if:
    \begin{equation*}
        \forall o \in O_M : [(x, g_o(x)) \in R_N \Leftarrow (f(x), o) \in S_M],
    \end{equation*}
    where $f(x) = f_1(x)f_2(x)\ldots f_{l_M}(x)$. The \textit{depth} $d$ of the pseudo-reduction is the maximum depth of decision trees in $f$ and $g$. The \textit{complexity} of the pseudo-reduction is $\max(d, \log M)$. A \textit{reduction} from $R$ to $S$ is a sequence of pseudo-reductions from $R_N$ to $S_M$ for $M = M(N)$, such that $N$-th pseudo-reduction is successful on all inputs of $R_N$ and has complexity $\poly(\log N)$.
\end{definition}

Note that pseudo-reduction is not an established term. We introduce it to explicitly distinguish the cases where we have no guarantee that $(f,g)$ has small complexity and is successful on all inputs.

\begin{definition}
    Given search problems $R_N$ and $S_M$, a \textit{randomized pseudo-reduction} from $R_N$ to $S_M$ is a distribution $D$ over some finite set of pseudo-reductions from $R_N$ to $S_M$. The \textit{success rate} on input $x$ is:
    \begin{equation*}
        \mathbb{P}_{(f,g)\sim D} [\forall o \in O_M : (x, g_o(x)) \in R_N \Leftarrow (f(x), o) \in S_M].
    \end{equation*}
    The depth of randomized pseudo-reduction is equal to the maximum depth of pseudo-reductions in $D$. Complexity is defined in the same way. A \textit{randomized reduction} from $R$ to $S$ is a sequence of randomized pseudo-reductions from $R_N$ to $S_M$ for $M = M(N)$, such that $N$-th randomized pseudo-reduction has a success rate of at least $\frac{1}{\poly(\log N)}$ and complexity $\poly(\log N)$.
\end{definition}

This definition can be seen as one-round communication with an adversary, who, given an arbitrary $x$ and sampled $(f,g)$, tries to provide worst possible solution for $f(x)$. Note that the adversary knows which $(f,g)$ was sampled, which makes it similar to public coin communication. As far as we know, this definition of randomized reduction between $\TFNP$ problems is most common, so we are going to use it. Nonetheless, it would be interesting to see if this result holds for private coin definition.

Also note that some definitions of randomized reductions require a higher success rate threshold. We have chosen one with a weaker requirements. Since we are going to show that there is no randomized reduction, our result will also work for more demanding definitions.

We are interested in proving that there is no randomized reduction. For this purpose, it is more convenient to use an equivalence, which follows from Yao's minimax principle \cite{Yao77}.

\begin{lemma}\label{lemma:randomized-reducibility-alt}
    There is a randomized reduction from a search problem $R$ to a search problem $S$ if and only if there are $p(N) = \poly(\log N)$ and $M = M(N)$ such that for any $N$ and distribution $\mathcal{X}$ over input of $R_N$ there is a pseudo-reduction $(f,g)$ from $R_N$ to $S_M$ with complexity at most $p(N)$ and success rate at least $\frac{1}{p(N)}$. 
\end{lemma}

Finally, sometimes we limit our search problem $R$ to some family of inputs $X = \{X_N \subset \Sigma_N^{l_N}\}_{N \in \mathbb{N}}$. We denote it as $(R,X)$. In that case, for any pseudo-reduction $(f,g)$ from $(R_N,X_N)$ to $(S_M,Y_M)$ we require that $x\in X \Rightarrow f(x) \in Y$; also, for a reduction, we require success only on $x \in X_N$.

\subsection{Partial Assignment}

\begin{definition}
    A \textit{partial assignment} on inputs of the search problem $R_N$ is any element $\rho \in (\Sigma_N \cup \{*\})^{l_N}$. The \textit{size} of $\rho$ is $|\rho| = |\{i \mid \rho_i \neq *\}|$. We say that $\rho_1$ \textit{extends} $\rho_2$ (denoted $\rho_1 \sqsupset \rho_2$) if they agree in all non-$*$ positions of $\rho_2$. Partial assignments $\rho_1$ and $\rho_2$ are \textit{consistent} (denoted $\rho_1 \agree{} \rho_2$) if there is an input that extends both of them; \textit{consistent relative to} $X_N \subset \Sigma_N^{l_N}$ ($\rho_1 \agree{X_N} \rho_2$) if such an input exists in $X_N$. 
\end{definition}

Partial assignments represent information we know about an input of a search problem. They are particularly useful when we are working with decision trees. A partial assignment naturally corresponds to every path from the root to a leaf in a decision tree. This partial assignment consists of answers to all queries made on this path of the decision tree.

\begin{definition}
    Given the search problem $R_N$ and a partial assignment $\rho$, we say that $\rho$ \textit{witnesses} a solution $o \in O_N$ if $(x,o) \in R_N$ for all inputs $x$ extending $\rho$. A partial assignment $\rho$ is called \textit{witnessing} if it witnesses some solution.
\end{definition}

\subsection{Properties of Pseudo-Reductions} \label{subsection:reduction-properties}
In this section, we introduce definitions for honest and perceptive pseudo-reductions. To the best of our knowledge, we are the first to introduce these properties. They are useful for us because they simultaneously are easy to enforce and guarantee a nice structure for a set of inputs on which a pseudo-reduction is not successful. 

Without loss of generality, we assume that all studied search problems contain a dedicated always incorrect solution $\bot$. This allows a pseudo-reduction to explicitly show that it does not know a correct solution.

\begin{definition}
    We call a pseudo-reduction $(f,g)$ from $R_N$ to $S_M$ \textit{honest} if, for all $o' \in O_M$ and $x \in \Sigma_N^{l_N}$, either $g_{o'}(x) = \bot$ or the partial assignment corresponding to the computation of $g_{o'}$ on $x$ witnesses the solution $g_{o'}(x)$.
\end{definition}

Intuitively, this property means that whenever the pseudo-reduction is unable to find a correct solution, it says so instead of trying to guess it. Any pseudo-reduction can be transformed into an honest one with an overhead of $\poly(\log N)$ in depth. To achieve this, for every $o'\in O_M$ and for every leaf in $g_{o'}$, we modify $g_{o'}$ so that it additionally checks if the solution in this leaf is correct, and if not, it returns $\bot$ instead. Since $R \in \TFNP$, this check can be made by a decision tree of $\poly(\log N)$ depth. Note that if the pseudo-reduction was successful on an input $x$, it will remain so.

\begin{definition}
Let $(f,g)$ be a pseudo-reduction from $R_N$ to $S_M$, where $R,S \in \TFNP$. Let us fix some family of decision trees $\mathcal{T}_M = \{T_{M,o}\}_{o\in O_M}$ that checks the correctness of solutions for $S_M$. We call a pseudo-reduction $(f,g)$ \textit{perceptive} (with respect to $\mathcal{T}_M$) if:
\begin{enumerate}
    \item For any leaf $l$ in $g_o$, if the corresponding partial assignment is witnessing, then $l$ returns any witnessed solution.
    \item Let $\overline{T}_{M,o}$ be the decision tree over the inputs of $R_N$, constructed from $T_{M,o}$ by substituting every query to $i$-th position of an input of $S_M$ with the tree $f_i$. We require that $g_o$ contains $\overline{T}_{M,o}$ in the sense that any path in $g_o$ as a partial assignment extends some path in $\overline{T}_{M,o}$.
\end{enumerate}
\end{definition}
This is complementary to the honesty property, which intuitively means that if a pseudo-reduction has enough information to give a correct solution, it does so.

We will usually omit the choice of the family of decision trees that checks the correctness of a solution. All we want from it is to check accurately and to have polylogarithmic depth. With respect to such a family, any pseudo-reduction can be transformed into a perceptive one with at most a polylogarithmic blowup in size. To achieve this, we apply the following operations:
\begin{enumerate}
    \item We extend every $g_o$ by replacing all leaves with the decision tree $\overline{T}_{M,o}$. New leaves inherit the value of the replaced leaf.
    \item In every leaf in the new decision trees, if the corresponding path is witnessing, we replace the leaf's value with any witnessed solution.
\end{enumerate} 
Note that if the pseudo-reduction was successful on an input $x$, it will remain so. Moreover, this operation preserves honesty.

Now we are ready to show why these properties are useful. For a pseudo-reduction of $R_N$ to $S_M$, we call an input of $R_N$ \textit{bad} if the pseudo-reduction is not successful on it. The following lemma shows how such inputs are distributed for an honest and perceptive pseudo-reduction.
\begin{lemma}
    \label{lemma:bad-input-structure}
    Let $(f,g)$ be an honest and perceptive pseudo-reduction of $R_N$ to $S_M$ of depth $d$. Then, for any bad input $x$, there is a non-witnessing partial assignment $\rho$ such that:
    \begin{itemize}
        \item $|\rho| \leq d$,
        \item $\rho \sqsubset x$,
        \item Any $x'$ such that $\rho \sqsubset x'$ is also bad.
    \end{itemize}
\end{lemma}

\begin{proof}
    Since $x$ is bad, there is $o \in O_M$ such that $(f(x), o)\in S_M$, but $(x, g_o(x)) \notin R_N$. Let $\rho$ be the partial assignment corresponding to the execution of $g_o$ on $x$. It is clear to see that $|\rho| \leq d$ and $\rho \sqsubset x$. Since the pseudo-reduction is perceptive and $(x, g_o(x)) \notin R_N$, $\rho$ must be non-witnessing. Also, since our pseudo-reduction is honest, $g_o(x)=\bot$. Then, for any $x' \sqsupset \rho$, we have $(x', g_o(x')) = (x', \bot) \notin R_N$. Lastly, $(f(x'),o)\in S_M$ since $T_{M,o}(f(x'))=T_{M,o}(f(x))=1$ due to perceptiveness.
\end{proof}

\begin{corollary}
    \label{corollary:bad-input-structure}
    We can apply this lemma to every bad input and combine the resulting $\rho$ into one set $B$. Then $B$ is a set of non-witnessing partial assignments with sizes at most $d$ such that $x$ is bad if and only if it extends some $\rho \in B$.
\end{corollary}

\subsection{Search Problem Restriction}

\begin{definition}
    Given a search problem $R_N$, a set of inputs $X_N$, and a partial assignment $\rho$, a \textit{restriction} of the pair $(R_N, X_N)$ by $\rho$ (denoted as $(R_N, X_N)\restrharpoon \rho$) is a pair $(R'_N, X'_N)$, where $\Sigma'_N = \Sigma_N$, $l'_N = l_N - |\rho|$, $O'_N=O_N$, $X'_N \subset \Sigma_N^{l'_N}$ is a set of strings such that their substitution in $\rho$ in place of all $*$ gives an element of $X_N$, $R'_N = \{(x',o) \mid$ substituting $x'$ into $\rho$ gives $x$ such that $(x,o) \in R_N\}$.
\end{definition}

If we have a pseudo-reduction $(f,g)$ from $(R_N, X_N)$ to a search problem $S_M$, then we can also apply the restriction to $(f,g)$ (denoted as $(f,g)\restrharpoon \rho$) to get a pseudo-reduction from $(R_N, X_N)\restrharpoon \rho$ to $S_M$. It is achieved simply by resolving in $(f,g)$ all queries that are already specified in $\rho$. Moreover, the following relation holds:

\begin{lemma}\label{lemma:restriction-property}
    Given a pseudo-reduction $(f,g)$ from $(R_N, X_N)$ to a search problem $S_M$ and a partial assignment $\rho$, if $x' \in X_N \restrharpoon \rho$ is a bad input for pseudo-reduction $(f,g)\restrharpoon \rho$, then $x \in X_N$, obtained by substituting $x'$ into $\rho$, is bad for pseudo-reduction $(f,g)$.
\end{lemma}

\begin{proof}
    Since $x'$ is bad, there is $o$ such that $((f\restrharpoon \rho)(x'), o) \in S_M$, $(x', (g_o\restrharpoon\rho)(x')) \notin R_N\restrharpoon \rho$. But $(f\restrharpoon \rho)(x') = f(x)$, which means $(f(x), o) \in S_M$; and $(g_o\restrharpoon\rho)(x') = g_o(x)$, which means $(x',g_o(x)) \notin R_N\restrharpoon \rho \Leftrightarrow (x,g_o(x)) \notin R_N$.
\end{proof}

\section{Main Theorem}

\begin{theorem}[Main theorem]\label{theorem:main}
    Let $R, S \in \TFNP$, $X = \{X_N \subset \Sigma_N^{l_N}\}_{N \in \mathbb{N}}$. Let the following conditions be met:
    \begin{enumerate}
        \item \label{item:main-nonreducibility}
        There is no reduction from $(R,X)$ to $S$.
        \item For any $p(N)=\poly(\log N)$ and for any sufficiently large $N$, there are sets $\mathcal{P}_1 = \mathcal{P}_1(N)$ and $\mathcal{P}_2 = \mathcal{P}_2(N)$ of partial assignments over inputs of $R_N$ such that:
        \begin{enumerate}
            \item \label{item:main-extendability} For any non-witnessing partial assignment $\rho$ over inputs of $R_N$ with $|\rho| \leq p(N)$ and such that $\rho$ is extendable to some $x \in X_N$, there is $A_\rho \subset \mathcal{P}_1$ such that $\forall x\in X_N$ $\left(\rho \sqsubset x \Leftrightarrow \exists \tau\in A_\rho: \tau \sqsubset x\right)$.
            \item \label{item:main-isomorphysm} For any sequence $\{\kappa_N \in \mathcal{P}_2(N)\}_N$, $(R,X)$ is reducible to $\{(R_N,X_N)\restrharpoon \kappa_N\}_N$.
            \item \label{item:distibutions} For $\kappa$ taken uniformly from $\mathcal{P}_2$, and then $x$ taken uniformly from $\{x\in X_N:x \sqsupset \kappa\}$, the resulting distribution of $x$ is uniform on $X_N$.
            \item \label{item:main-extendability-vs-consistency}There is $\delta > 0$ such that for any $T \subset \mathcal{P}_1$:
            \begin{gather*}
                |\{\kappa\in\mathcal{P}_2 \mid \exists \tau \in T: \kappa \sqsupset \tau\}| \geq (1 - o(N^{-\delta}))|\{\kappa\in\mathcal{P}_2\mid \exists \tau \in T:  \kappa \agree{X_N} \tau\}|.
            \end{gather*}
        \end{enumerate}
    \end{enumerate}
    Then there is no randomized reduction from $(R,X)$ to $S$.
\end{theorem}

Before we proceed with the proof, let us informally explain what sets $\mathcal{P}_1$ and $\mathcal{P}_2$ represent and how this proof works. 

The set $\mathcal{P}_1$ consists of \enquote{small} partial assignments. Usually we want to choose $\mathcal{P}_1$ so that all its elements are non-witnessing and have the same size and structure. The purpose of this set is to refine the set of partial assignments, resulting from the application of \autoref{corollary:bad-input-structure}, into the set $T \subset \mathcal{P}_1$ with the same properties. This allows us to work with small partial assignments of the same size and structure, which is convenient.

The set $\mathcal{P}_2$ consists of \enquote{big} partial assignments. As with $\mathcal{P}_1$, we usually want to choose $\mathcal{P}_2$ so that all its elements are non-witnessing and have the same size and structure. Moreover, from $\kappa \in \mathcal{P}_2$ we expect that restricting $(R_N,X_N)$ by $\kappa$ turns it into essentially the same problem, but smaller.

For a quick example of such $\mathcal{P}_1$, $\mathcal{P}_2$ and $X_N$ we refer to \autoref{diagram:pigeon-structure}.

In general, the proof is structured as follows:
\begin{itemize}
    \item We are given $(f,g)$ with small complexity and we want to show that it performs poorly under the uniform distribution over $X_N$.
    \item Assuming $(f,g)$ is honest and perceptive, apply \autoref{corollary:bad-input-structure} to obtain the set $B$ of partial assignments. For $B$ holds that $x \in X_N$ is bad iff $\exists \rho \in B: x \sqsupset \rho$. 
    \item Use condition \ref{item:main-extendability} to refine $B$ into $T \subset \mathcal{P}_1$ with the same property as $B$.
    \item Show that (for infinitely many $N$) any $\kappa \in \mathcal{P}_2$ must be consistent (relative to $X_N$) with some $\tau \in T$. Indeed, otherwise, we can restrict our pseudo-reduction by $\kappa$ to kill all $\tau \in T$, which means that the resulting pseudo-reduction is always successful. But that would contradict condition \ref{item:main-nonreducibility}, since a restriction by $\kappa$ turns our problem into essentially the same problem, but smaller.
    \item Use condition \ref{item:main-extendability-vs-consistency} to show that almost all $\kappa \in \mathcal{P}_2$ must extend some $\tau \in T$.
    \item Use condition \ref{item:distibutions} to show that the same holds for almost all $x \in X_N$, which due to the property of $T$ means that almost all $x \in X_N$ are bad. This concludes the proof.
\end{itemize}

\begin{proof}
    Let us fix any $p(N) = \poly(\log N)$. For arbitrary $N$ let $(f,g)$ be any pseudo-reduction of $R_N$ to $S_M$ with complexity $p(N)$. Our goal is to show that this pseudo-reduction performs poorly on $x$, taken uniformly from $X_N$. We can assume $(f,g)$ to be honest and perceptive: as shown in \autoref{subsection:reduction-properties}, these properties can be achieved with polynomial blowup in complexity. Moreover, if the pseudo-reduction was successful on an input, it will remain so. Then, from \autoref{corollary:bad-input-structure}, there is a set of non-witnessing partial assignments $B$ such that $x\in X_N$ is bad if and only if it extends some $\rho \in B$. We exclude from $B$ all $\rho$ that are not extendable to any $x \in X_N$. Since all $\rho \in B$ have a size at most $p(N)$, we can use condition $\ref{item:main-extendability}$ to construct $T = \bigcup_{\rho \in B} A_\rho \subset \mathcal{P}_1$. By design, for such $T$ we also have that $x\in X_N$ is bad if and only if it extends some $\tau \in T$.
    
    Now, we proceed to show that for infinitely many $N$, any $\kappa\in \mathcal{P}_2$ is consistent relative to $X_N$ with at least one $\tau \in T$. Assume that there is $\kappa$ such that it is not consistent relative to $X_N$ with any $\tau \in T$. Then the pseudo-reduction $(f_N,g_N)\restrharpoon \kappa$ from $(R_N,X_N)\restrharpoon\kappa_N$ to $S_M$ is successful on all inputs of $(R_N,X_N)\restrharpoon\kappa_N$. Indeed, if $x'$ is a bad input, from \autoref{lemma:restriction-property} we have that $x$, obtained by substituting $x'$ into $\kappa$, is also bad. But it is impossible, since a bad input must extend some $\tau \in T$, which would contradict inconsistency of $\kappa$. Now, if such $\kappa$ exists for all sufficiently large $N$, then $\{(R_N,X_N)\restrharpoon \kappa_N\}_N$ is reducible to $S$, which together with conditions \ref{item:main-isomorphysm} and \ref{item:main-nonreducibility} leads to a contradiction.

    Finally, from condition \ref{item:main-extendability-vs-consistency} we can conclude that for some $\delta > 0$ for infinitely many $N$ at least $1 - o(N^{-\delta})$ of all $\kappa \in \mathcal{P}_2$ are not just consistent, but even extend some $\tau \in T$. Since from condition \ref{item:distibutions} random $x$ can be generated by extending random $\kappa \in \mathcal{P}_2$, at least $1 - o(N^{-\delta})$ of all $x \in X_N$ also extend some $\tau\in T$, which means that they are bad. So, at most $o(N^{-\delta})$ of all inputs are good. But in \autoref{lemma:randomized-reducibility-alt} for randomized reduction to exist we require a success rate of at least $\frac{1}{\poly(\log N)}$.
\end{proof}
\subsection{Auxiliary Results}
The following results are useful for showing that the premise of the main theorem holds. In particular, the following lemma shows that under certain conditions on $R \in \TFNP$ and its inputs $X_1$ and $X_2$, for any $S \in \TFNP$, $(R,X_1)$ is reducible to $S$ if and only if $(R,X_2)$ is reducible to $S$. This is useful to get condition \ref{item:main-nonreducibility} of the theorem.

\begin{lemma} \label{lemma: dence-input-subset}
    Let $R = (R_N)$ be any search problem from $\TFNP$, let $X$ and $X'$ be inputs of $R$ such that $X_N \subset X'_N$ for all $N$. Assume that for any $p(N) = \poly(\log N)$ for all sufficiently large $N$ for any non-witnessing partial assignment $\rho$ of size at most $p(N)$ it holds that if $\rho$ can be extended to $x' \in X'_N$, it can also be extended to $x \in X_N$. Then, for any $S \in \TFNP$, $(R, X)$ is reducible to $S$ if and only if $(R, X')$ is reducible to $S$.
\end{lemma}

\begin{proof}
    It suffices to show that if $(R,X')$ is not reducible to $S$, then  $(R,X)$ is also not reducible to $S$. In addition, we can consider only honest and perceptive pseudo-reductions. Let $p(N) = \poly(\log N)$, let $(f,g)$ be an arbitrary honest and perceptive pseudo-reduction from $(R_N, X')$ to $S_M$ with the complexity at most $p(N)$. For infinitely many $N$, there is a bad input $x' \in X'_N$ for any such pseudo-reduction. Then we can take the partial assignment $\rho \sqsubset x'$ from \autoref{lemma:bad-input-structure}. For sufficiently large $N$, $\rho$ can be extended to $x \in X_N$, which means that $x$ is also bad. That concludes the proof.
\end{proof}

If we take $\mathcal{P}_2$ and $X_N$ as in \autoref{theorem:main}, we can associate with them graph $G$ with $V(G)=\mathcal{P}_2$, and $E(G)$ being pairs of elements from $\mathcal{P}_2$ that are consistent relative to $X_N$. If $\mathcal{P}_2$ and $X_N$ are very symmetric, this graph has useful properties. These properties are demonstrated in the following lemma:

\begin{lemma}\label{lemma:good-graph}
    Let $G$ be a graph, $V_0 \sqcup \ldots \sqcup V_n \subset V(G)$. Let there be $d_{0\to k}, d_{k \to 0} > 0$ for all $k \in [n]$ such that:
    \begin{itemize}
        \item $\forall v \in V_0: |N(v) \cap V_k| = d_{0\to k}$;
        \item $\forall v \in V_k: |N(v) \cap V_0| = d_{k\to 0}$.
    \end{itemize}
    There $N(\cdot)$ is the neighborhood function. Also, let $N(V_0) \subset V_0 \sqcup \ldots \sqcup V_n$. Then for any nonempty $V' \subset V_0$, it holds:
    \begin{equation*}
        \frac{|N(V')\setminus V_0|}{|V'|} \geq \frac{|N(V_0)\setminus V_0|}{|V_0|}
    \end{equation*}
\end{lemma}
\begin{proof}
    \begin{gather*}
        |N(V')\setminus V_0| = \sum_{k \in [n]} |N(V')\cap V_k| \geq \sum_{k\in [n]} |V'|\frac{d_{0\to k}}{d_{k\to 0}} = |V'|\sum_{k\in [n]} \frac{d_{0\to k}}{d_{k\to 0}}\\
        |N(V_0)\setminus V_0| = \sum_{k \in [n]} |N(V_0)\cap V_k| = \sum_{k\in [n]} |V_0|\frac{d_{0\to k}}{d_{k\to 0}} = |V_0|\sum_{k\in [n]} \frac{d_{0\to k}}{d_{k\to 0}}\\
         \frac{|N(V')\setminus V_0|}{|V'|} \geq \sum_{k\in [n]} \frac{d_{0\to k}}{d_{k\to 0}} = \frac{|N(V_0)\setminus V_0|}{|V_0|}.
    \end{gather*}
\end{proof}

For nonempty $T \subset \mathcal{P}_1$ let us look at the ratio between the number of $\kappa \in \mathcal{P}_2$ that are consistent with some $\tau \in T$ relative to $X_N$ and the number of $\kappa \in \mathcal{P}_2$ that extend some $\tau \in T$. The condition \ref{item:main-extendability-vs-consistency} of \autoref{theorem:main} requires this ratio to be sufficiently close to $1$. The following lemma uses the previous result to show that under certain conditions, if we already showed that this ratio is close to 1 for $|T|=1$, the same holds for arbitrary $T$.

\begin{lemma}\label{lemma:tiling-optimality}
    Let $X_N \subset \Sigma_N^{l_N}$, let $\mathcal{P}_1$ and $\mathcal{P}_2$ be sets of partial assignments over $X_N$. Let it be true for some $\gamma > 1$ that for every $\tau \in \mathcal{P}_1$, it holds:
    \begin{enumerate}
        \item For every $x \in X_N$ such that $\tau \sqsubset x$, there is $\kappa \in \mathcal{P}_2$ such that $\tau \sqsubset \kappa \sqsubset x$.
        \item $|\{\kappa \in \mathcal{P}_2 \mid \kappa \agree{X_N} \tau\}| = \gamma|\{\kappa \in \mathcal{P}_2 \mid \kappa \sqsupset \tau\}|$
        \item For graph $G$ with $V(G) = \mathcal{P}_2$, $E(G) = \{(\kappa_1, \kappa_2) \mid \kappa_1 \agree{X_N} \kappa_2\}$, $V_0 = \{\kappa \in \mathcal{P}_2 \mid \kappa \sqsupset \tau\}$,  $V_k = \{\kappa \in \mathcal{P}_2 \mid \kappa \agree{X_N} \tau, \text{ $\kappa$ and $\tau$ have $|\tau|-k$ common non-$*$ positions} \}$, all conditions of \autoref{lemma:good-graph} are met.
    \end{enumerate}
    Then, for any $T \subset \mathcal{P}_1$, it holds that:
    \begin{gather*}
        |\{\kappa \in \mathcal{P}_2 \mid  \exists \tau \in T: \kappa \agree{X_N} \tau\}| \leq \gamma|\{\kappa \in \mathcal{P}_2 \mid \exists \tau \in T: \kappa \sqsupset \tau\}|.
    \end{gather*}
\end{lemma}
\begin{proof}
    The first condition gives us the following equivalence:
    \begin{claim}
    \label{claim:alternative-view-on-consistency}
        For any $\tau \in \mathcal{P}_1$ and $\kappa \in \mathcal{P}_2$:  $\kappa \agree{X_N} \tau$ $\Leftrightarrow$ $\exists \kappa' \in \mathcal{P}_2: \kappa' \sqsupset \tau \land \kappa\agree{X_N}\kappa'$. 
    \end{claim}
    
    Now, let us prove the claim of the lemma by induction on the size of $T$. For $|T| = 1$, we already know that the equality holds. Let $T = T_0 \cup \{\tau\}$ and for $T_0$ the required inequality holds. Define $F(\tau) = \{\kappa \in \mathcal{P}_2 \mid \kappa \sqsupset \tau\}$, $F_w(\tau) = \{\kappa \in \mathcal{P}_2 \mid \kappa \agree{X_N} \tau\}$, $N(\kappa') = \{\kappa \in \mathcal{P}_2 \mid \kappa \agree{X_N} \kappa'\}$. For sets of $\tau$ or $\kappa$ definitions are similar. Note that $F_w = N\circ F$ due to \autoref{claim:alternative-view-on-consistency}. In these terms, we need to prove that:
    \begin{gather*}
        |F_w(T)| \leq \gamma |F(T)|.
    \end{gather*}
    First, assume the case $F(T_0) \cap F(\tau) = \emptyset$. Then:
    \begin{gather*}
        |F_w(T)| \leq |F_w(T_0)| + |F_w(\tau)| \leq \gamma(|F(T_0)| + |F(\tau)|) = \gamma |F(T)|.
    \end{gather*}
    Now assume that $F(T_0) \cap F(\tau) = V' \neq \emptyset$, $F(\tau) \setminus V' = V''$. Then we have:
    \begin{align*}
        &|F_w(T)| = \\
        &|F_w(T_0)| + |F_w(\tau)\setminus F_w(T_0)| \leq && \text{(splitting $F_w(T)$ into 2 parts)}\\ 
        &\gamma|F(T_0)| + |F_w(\tau)\setminus F_w(T_0)| = && \text{(applying induction)}\\
        &\gamma|F(T_0)| + |F_w(\tau)\setminus N(F(T_0))| \leq && \text{(using that $F_w = N \circ F$)}\\
        &\gamma|F(T_0)| + |F_w(\tau)\setminus N(V')| = && \text{(using that $V' \subset F(T_0)$)}\\
        &\gamma|F(T_0)| + |F_w(\tau)| - |N(V')| = && \text{($F_w(\tau) = N(F(\tau)) \supset N(V')$)}\\
        &\gamma|F(T_0)| + \gamma|F(\tau)| - |N(V')| = && \text{(applying the second condition)}\\
        &\gamma|F(T_0)| + \gamma|V''| + \gamma|V'| - |N(V')| = && \text{($F(\tau) = V' \sqcup V''$)}\\
        &\gamma|F(T)| + \gamma|V'| - |N(V')|. && \text{($F(T) = F(T_0) \sqcup V''$)}
    \end{align*}
    Finally, $\gamma|F(T)| + \gamma|V'| - |N(V')| \leq \gamma|F(T)|$, since:
    \begin{align*}
        &|N(V')| \geq \\
        &|V'| + |N(V')\setminus F(\tau)| \geq &&\text{($V' \subset N(V')$ and $V' \subset F(\tau)$)}\\
        &|V'| + \frac{|V'|}{|F(\tau)|}|N(F(\tau))\setminus F(\tau)| = &&\text{(using \autoref{lemma:good-graph} with $V_0 = F(\tau)$)}\\
        &\frac{|N(F(\tau))|}{|F(\tau)|}|V'| = &&\text{($F(\tau) \subset N(F(\tau))$)}\\
        &\gamma|V'|. && \text{(applying the second condition)}
    \end{align*}
\end{proof}

\section{Applications of the Main Theorem}
\label{section:application}
\subsection{PPP}
\begin{definition}
    $\PPP$ is a class of $\TFNP$ problems that are reducible to the search problem $\Pigeon$.\\
    $\Pigeon = \{\Pigeon_N \subset \Sigma_N^{l_N}\times O_N\}$, where:\\
    $\Sigma_N = [N-1]\cup\{\bot\}$, $l_N=N$, $O_N = [N]\cup\{(i,j)\in[N]^2\mid i\neq j\}$.\\
    Positions $i \in [N]$ are called pigeons, and values $j \in [N-1]$ are called holes, $x_i = j$ represents pigeon $i$ being mapped to hole $j$.\\
    Correct solutions for input $x$: any $i$ such that $x_i = \bot$ (homeless pigeon) or $(i,j)$ such that $x_i = x_j$ (collision).
\end{definition}

Our goal is to show that there is no randomized reduction from $\Pigeon$ to $S \in \TFNP$, given that there is no deterministic reduction from $\Pigeon$ to $S \in \TFNP$. To apply \autoref{theorem:main}, we first need to come up with a set of inputs $X_N$ for each $N$. We take $X_N$ to be the set of all bijections from $[N]$ to $[N-1]\cup\{\bot\}$. Now we can apply \autoref{lemma: dence-input-subset} (with $X'$ being all inputs) to show that $(\Pigeon, X)$ is also not reducible to $S$. The conditions of the lemma hold, since any non-witnessing partial assignment contains neither a homeless pigeon nor a collision, which means it can be extended to $x \in X_N$.

Now, given $p(N)=\poly(\log N)$, we need to come up with $\mathcal{P}_1$ and $\mathcal{P}_2$ for every sufficiently large $N$. As $\mathcal{P}_1$ and $\mathcal{P}_2$ we take sets of all non-witnessing partial assignments of size $r$ and $k$, respectively (see \autoref{diagram:pigeon-structure}). There $r = p(N)$, $k = N - \floor{N^{\eps}}$ for $\eps = \frac 14$. Let us go through the sub-items of \autoref{theorem:main}:
\begin{itemize}
    \item \ref{item:main-extendability}: For any non-witnessing partial assignment $\rho$ with $|\rho| \leq p(N)$ we can take $A_\rho = \{\rho' \in \mathcal{P}_1 \mid \rho' \sqsupset \rho\}$. It is clear that $x$ extends $\rho$ if and only if it extends some $\rho' \in A_\rho$.
    \item \ref{item:main-isomorphysm}: Note that for every $N$ there is $\hat{N} = O(N^{\frac{1}{\eps}})$ such that $\floor{\hat{N}^{\eps}} = N$. Also, for any $\kappa \in \mathcal{P}_2(\hat{N})$, the search problem $(\Pigeon_{\hat{N}}, X_{\hat{N}})\restrharpoon \kappa$ is equivalent to $(\Pigeon_N, X_N)$ up to an alphabet renaming, which means there is a trivial always successful pseudo-reduction from $(\Pigeon_N, X_N)$ to $(\Pigeon_{\hat{N}}, X_{\hat{N}})\restrharpoon \kappa$, with complexity $\log \hat{N} = O(\log N)$. From this follows the required reducibility.
    \item \ref{item:distibutions}: By the symmetry of $\mathcal{P}_2$ and $X_N$.
    \item \ref{item:main-extendability-vs-consistency}: For this one, we can use \autoref{lemma:tiling-optimality}. Note that due to symmetry the value $\gamma \coloneqq |\{\kappa \in \mathcal{P}_2 \mid  \kappa \agree{X_N} \tau\}|/|\{\kappa \in \mathcal{P}_2 \mid \kappa \sqsupset \tau\}|$ does not depend on the choice of $\tau \in \mathcal{P}_1$. Other conditions of the lemma also hold. Finally, simple combinatorial estimates (see \autoref{appendix:pigeon-calculations}) show that $\gamma = (1 + O(N^{-\frac 14})) \Rightarrow \gamma^{-1} = (1 - o(N^{-\frac 15}))$, which concludes the proof.
\end{itemize}

\begin{figure}[ht]
    \centering
    \captionsetup[subfigure]{justification=centering}
    \begin{subfigure}{0.3\textwidth}
        \centering
        \begin{tikzpicture}[node distance=1cm]
            \node[square] (a11){};
            \node[circ,above of=a11] (a21){};
            \node[square,right of=a11] (a12){};
            \node[circ] (a22) at (a21 -| a12) {};
            \draw[decoration={brace,mirror,raise=6pt},decorate](a11.west) -- node[below=6pt] {$r$} (a12.east);
            \draw[arrow](a11) -- (a21);
            \draw[arrow](a12) -- (a22);
            \node at ($(a11)!0.5!(a12)$) {$\ldots$};
            \node at ($(a21)!0.5!(a22)$) {$\ldots$};

            \node[square,right of=a12,node distance=0.5cm] (a13) {};
            \node[above of=a13,node distance=\shift1] {$*$};
            \node[square,right of=a13] (a14) {};
            \node[above of=a14,node distance=\shift1] {$*$};
            \node at ($(a13)!0.5!(a14)$) {$\ldots$};
            \node[square,right of=a14,node distance=0.5cm] (a15) {};
            \node[above of=a15,node distance=\shift1] {$*$};
            \draw[decoration={brace,mirror,raise=6pt},decorate](a13.west) -- node[below=6pt] {$N-r$} (a15.east);

            \node[circ] (a23) at (a22 -| a13) {};
            \node[circ] (a24) at (a23 -| a14) {};
            \node at ($(a23)!0.5!(a24)$) {$\ldots$};
            \draw[decoration={brace,raise=6pt},decorate](a23.west) -- node[above=6pt] {$N-1-r$} (a24.east);
        \end{tikzpicture}
        \caption{The structure of $\tau \in \mathcal{P}_1$.}
    \end{subfigure}
    \hfill
    \begin{subfigure}{0.3\textwidth}
        \centering
        \begin{tikzpicture}[node distance=1cm]
            \node[square] (a11){};
            \node[circ,above of=a11] (a21){};
            \node[square,right of=a11] (a12){};
            \node[circ] (a22) at (a21 -| a12) {};
            \draw[decoration={brace,mirror,raise=6pt},decorate](a11.west) -- node[below=6pt] {$k$} (a12.east);
            \draw[arrow](a11) -- (a21);
            \draw[arrow](a12) -- (a22);
            \node at ($(a11)!0.5!(a12)$) {$\ldots$};
            \node at ($(a21)!0.5!(a22)$) {$\ldots$};

            \node[square,right of=a12,node distance=0.5cm] (a13) {};
            \node[above of=a13,node distance=\shift1] {$*$};
            \node[square,right of=a13] (a14) {};
            \node[above of=a14,node distance=\shift1] {$*$};
            \node at ($(a13)!0.5!(a14)$) {$\ldots$};
            \node[square,right of=a14,node distance=0.5cm] (a15) {};
            \node[above of=a15,node distance=\shift1] {$*$};
            \draw[decoration={brace,mirror,raise=6pt},decorate](a13.west) -- node[below=6pt] {$N-k$} (a15.east);

            \node[circ] (a23) at (a22 -| a13) {};
            \node[circ] (a24) at (a23 -| a14) {};
            \node at ($(a23)!0.5!(a24)$) {$\ldots$};
            \draw[decoration={brace,raise=6pt},decorate](a23.west) -- node[above=6pt] {$N-1-k$} (a24.east);
        \end{tikzpicture}
        \caption{The structure of $\kappa \in \mathcal{P}_2$.}
    \end{subfigure}
    \hfill
    \begin{subfigure}{0.3\textwidth}
        \centering
        \begin{tikzpicture}[node distance=1cm]
            \node[square] (a11){};
            \node[circ,above of=a11] (a21){};
            \node[square,right of=a11] (a12){};
            \node[circ] (a22) at (a21 -| a12) {};
            \node[square,right of=a12,node distance=0.5cm] (a13) {};
            \node[above of=a13,node distance=0.35cm] {$\bot$};
            \draw[decoration={brace,mirror,raise=6pt},decorate](a11.west) -- node[below=6pt] {$N-1$} (a12.east);
            \draw[arrow](a11) -- (a21);
            \draw[arrow](a12) -- (a22);
            \node at ($(a11)!0.5!(a12)$) {$\ldots$};
            \node at ($(a21)!0.5!(a22)$) {$\ldots$};
        \end{tikzpicture}
        \caption{The structure of $x \in X_N$.}
    \end{subfigure}
    \caption{The structure of elements in $\mathcal{P}_1$, $\mathcal{P}_2$ and $X_N$ for $\Pigeon_N$ up to a permutation of pigeons and holes. Here squares represent pigeons, circles represent holes, and arrows show how pigeons are mapped to holes.}
    \label{diagram:pigeon-structure}
\end{figure}
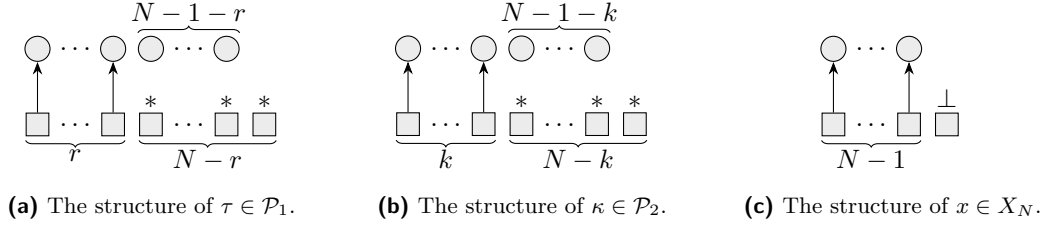

\subsection{PPAD and PPADS}
\begin{definition}
    $\PPAD$ and $\PPADS$ are classes of $\TFNP$ problems that are reducible to search problems $\InjPigeon$ (Injective Pigeon) and $\BijPigeon$ (Bijective Pigeon), respectively.\\
    $\InjPigeon = \{\InjPigeon_N \subset \Sigma_N^{l_N}\times O_N\}$, where:\\
    $\Sigma_N = [N]\cup\{\bot\}$, $l_N=2N-1$, $O_N = [N]$.\\
    Any $x \in \Sigma_N^{l_N}$ can be interpreted as a pair $(s,p) \in \Sigma_N^{N}\times \Sigma_N^{N-1}$. We count $i \in [N]$ as the correct solution for $x = (s,p)$ if either $s_i \in \{\bot, N\}$ or $p_{s_i} \neq i$. \\
    $\BijPigeon$ is defined in the same way, except that we also count $i \in [N]$ as the correct solution for $x = (s,p)$ if $i < N$ and either $p_i = \bot$ or $s_{p_i} \neq i$.
\end{definition}

The search problems $\InjPigeon$ and $\BijPigeon$ are introduced and shown to be complete in the corresponding classes in \cite{GHJ22}. For original definitions of $\PPAD$ and $\PPADS$, see \cite{Pap94}. These problems are similar to $\Pigeon$: there are $N$ pigeons and $N-1$ holes. The difference is that there is mapping not only from pigeons to holes but also from holes to pigeons. We say that pigeon $i$ is homeless, if $s_i \in \{N, \bot\}$; hole $j$ is empty if $p_j = \bot$, and pigeon-hole pair is a pair $(i,j)$ such that $s_i = j$, $p_j = i$. 

As with $\Pigeon$, our goal is to show the equivalence of deterministic and randomized reducibility from these problems to any $S \in \TFNP$. First, we restrict our attention to consistent inputs: $X'_N = \{x=(s,p) \mid (\forall i\in[N]$ $\forall j\in[N-1]: s_i = j \Leftrightarrow p_j = i) \land (\forall i\in [N] : s_i \neq N)\}$. It is easy to see that both of these problems are reducible to themselves, restricted to such inputs. From now on, we proceed with each problem separately.

\subsubsection{PPAD}

First, we can further refine $X'_N$ to $X_N$ so that there are no empty holes: $X_N = \{(s,p) \in X'_N \mid \forall j: p_j \neq \bot\}$. As with $\Pigeon$, \autoref{lemma: dence-input-subset} can be applied here: any non-witnessing $\rho$, extendable to $x' \in X'$, consists only of pigeon-hole pairs, which means it is extendable to $x \in X_N$.

Similarly to $\Pigeon$, as $\mathcal{P}_1$ and $\mathcal{P}_2$ we take sets of all non-witnessing partial assignments of size $2r$ and $2k$, respectively, with the additional requirement on partial assignments: $s_i = j \Leftrightarrow p_j = i$. There $r = p(N)$, $k = N - \floor{N^{\eps}}$ for $\eps = \frac 14$. Just like for $\Pigeon$, the elements of these sets represent partially defined injections from $[N]$ to $[N-1]$ with domain sizes $r$ and $k$, respectively. As it can be seen from \autoref{diagram:pigeon-structure} and \autoref{diagram:bij-pigeon-structure}, $\mathcal{P}_1, \mathcal{P}_2$ and $X_N$ for $\Pigeon_N$ and for $\BijPigeon_N$ look very similar. From now on, the proof almost exactly repeats the proof for $\Pigeon$, including combinatorial calculations. 

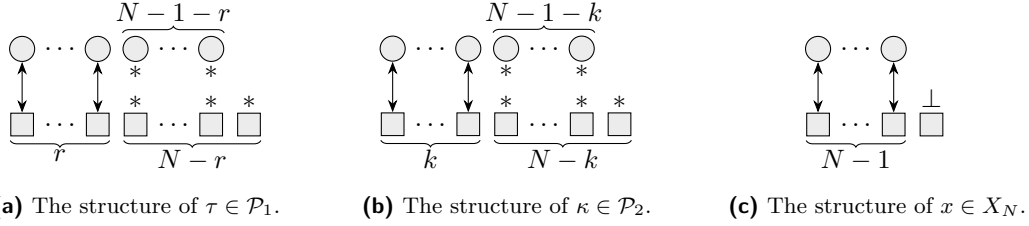
\begin{figure}[ht]
    \centering
    \captionsetup[subfigure]{justification=centering}
    \begin{subfigure}{0.3\textwidth}
        \centering
        \begin{tikzpicture}[node distance=1cm]
            \node[square] (a11){};
            \node[circ,above of=a11] (a21){};
            \node[square,right of=a11] (a12){};
            \node[circ] (a22) at (a21 -| a12) {};
            \draw[decoration={brace,mirror,raise=6pt},decorate](a11.west) -- node[below=6pt] {$r$} (a12.east);
            \draw[bi-arrow](a11) -- (a21);
            \draw[bi-arrow](a12) -- (a22);
            \node at ($(a11)!0.5!(a12)$) {$\ldots$};
            \node at ($(a21)!0.5!(a22)$) {$\ldots$};

            \node[square,right of=a12,node distance=0.5cm] (a13) {};
            \node[above of=a13,node distance=\shift1] {$*$};
            \node[square,right of=a13] (a14) {};
            \node[above of=a14,node distance=\shift1] {$*$};
            \node at ($(a13)!0.5!(a14)$) {$\ldots$};
            \node[square,right of=a14,node distance=0.5cm] (a15) {};
            \node[above of=a15,node distance=\shift1] {$*$};
            \draw[decoration={brace,mirror,raise=6pt},decorate](a13.west) -- node[below=6pt] {$N-r$} (a15.east);

            \node[circ] (a23) at (a22 -| a13) {};
            \node[below of=a23,node distance=\shift1] {$*$};
            \node[circ] (a24) at (a23 -| a14) {};
            \node[below of=a24,node distance=\shift1] {$*$};
            \node at ($(a23)!0.5!(a24)$) {$\ldots$};
            \draw[decoration={brace,raise=6pt},decorate](a23.west) -- node[above=6pt] {$N-1-r$} (a24.east);
        \end{tikzpicture}
        \caption{The structure of $\tau \in \mathcal{P}_1$.}
    \end{subfigure}
    \hfill
    \begin{subfigure}{0.3\textwidth}
        \centering
        \begin{tikzpicture}[node distance=1cm]
            \node[square] (a11){};
            \node[circ,above of=a11] (a21){};
            \node[square,right of=a11] (a12){};
            \node[circ] (a22) at (a21 -| a12) {};
            \draw[decoration={brace,mirror,raise=6pt},decorate](a11.west) -- node[below=6pt] {$k$} (a12.east);
            \draw[bi-arrow](a11) -- (a21);
            \draw[bi-arrow](a12) -- (a22);
            \node at ($(a11)!0.5!(a12)$) {$\ldots$};
            \node at ($(a21)!0.5!(a22)$) {$\ldots$};

            \node[square,right of=a12,node distance=0.5cm] (a13) {};
            \node[above of=a13,node distance=\shift1] {$*$};
            \node[square,right of=a13] (a14) {};
            \node[above of=a14,node distance=\shift1] {$*$};
            \node at ($(a13)!0.5!(a14)$) {$\ldots$};
            \node[square,right of=a14,node distance=0.5cm] (a15) {};
            \node[above of=a15,node distance=\shift1] {$*$};
            \draw[decoration={brace,mirror,raise=6pt},decorate](a13.west) -- node[below=6pt] {$N-k$} (a15.east);

            \node[circ] (a23) at (a22 -| a13) {};
            \node[below of=a23,node distance=\shift1] {$*$};
            \node[circ] (a24) at (a23 -| a14) {};
            \node[below of=a24,node distance=\shift1] {$*$};
            \node at ($(a23)!0.5!(a24)$) {$\ldots$};
            \draw[decoration={brace,raise=6pt},decorate](a23.west) -- node[above=6pt] {$N-1-k$} (a24.east);
        \end{tikzpicture}
        \caption{The structure of $\kappa \in \mathcal{P}_2$.}
    \end{subfigure}
    \hfill
    \begin{subfigure}{0.3\textwidth}
        \centering
        \begin{tikzpicture}[node distance=1cm]
            \node[square] (a11){};
            \node[circ,above of=a11] (a21){};
            \node[square,right of=a11] (a12){};
            \node[circ] (a22) at (a21 -| a12) {};
            \node[square,right of=a12,node distance=0.5cm] (a13) {};
            \node[above of=a13,node distance=0.35cm] {$\bot$};
            \draw[decoration={brace,mirror,raise=6pt},decorate](a11.west) -- node[below=6pt] {$N-1$} (a12.east);
            \draw[bi-arrow](a11) -- (a21);
            \draw[bi-arrow](a12) -- (a22);
            \node at ($(a11)!0.5!(a12)$) {$\ldots$};
            \node at ($(a21)!0.5!(a22)$) {$\ldots$};
        \end{tikzpicture}
        \caption{The structure of $x \in X_N$.}
    \end{subfigure}
    \caption{The structure of elements in $\mathcal{P}_1$, $\mathcal{P}_2$ and $X_N$ for $\BijPigeon_N$ up to a permutation of pigeons and holes. Here squares represent pigeons, circles represent holes, and arrows show how pigeons and holes are mapped to each other.}
    \label{diagram:bij-pigeon-structure}
\end{figure}

\subsubsection{Obstacles with PPADS}
It may seem that $\InjPigeon$ can be handled in the same way as $\BijPigeon$. However, unlike $\BijPigeon$, for $\InjPigeon$, partial assignments with empty holes do not witness an answer. From this it follows that if we want to apply \autoref{theorem:main}, we require that elements of sets $\mathcal{P}_1$, $\mathcal{P}_2$ and $X_N$ contain empty holes. It seems that this difference is significant, since we were unable to come up with such $\mathcal{P}_1$, $\mathcal{P}_2$ and $X_N$ that all conditions of \autoref{theorem:main} hold.

\subsection{t-PPP}

\begin{definition}
    For nondecreasing $t = t(N)$, $2 \leq t \leq \poly(\log N)$, $t$-$\PPP$ is a class of $\TFNP$ problems that are reducible to the search problem $t$-$\Pigeon = \{t(N)\text{-}\Pigeon_N\}_N$. $t\text{-}\Pigeon_N$ is defined as follows:\\
    $\Sigma_N = [N] \cup \{\bot\}$, $l_N = (t-1)N + 1$, $O_N = [l_N]\cup\{(i_1,\ldots,i_t) \in [l_N]^t\mid \text{$i_s$ are pairwise distinct}\}$.\\
    Correct solution for input $x$: any $i$ such that $x_i = \bot$ or $(i_1, \ldots, i_t)$ such that $x_{i_1} = \ldots = x_{i_t}$.
\end{definition}

These problems were introduced in \cite{Jain2024OnPP}. There, it was shown that for $n = \log N$, $t_1(n), t_2(n) = \poly(n)$, if $t_1(n) \leq t_2(n^c)$ and $t_2(n) \leq t_1(n^c)$ for some $c$ and all sufficiently large $N$, then $t_1$-$\Pigeon$ and $t_2$-$\Pigeon$ are reducible to each other. Using that, we may consider only $t(N) \leq \sqrt{\log N}$. Indeed, since $t(n) \leq n^c$ for some $c$ and sufficiently large $N$, $\hat{t}(n) = t(\floor{n^{\frac{1}{2c}}}) \leq \sqrt{\log N}$, and $\hat{t}(n) \leq t(n) \leq \hat{t}(n^{2c+1})$.

As a hard set of inputs for $t$-$\Pigeon_N$ we take $X_{N,t}$, which is a set of such inputs $x$ that $x$ equals $\bot$ in exactly one position, and for any $a \in [N]$, there are exactly $t-1$ positions of $x$ with that value. For such inputs, the only correct solution is the position of $\bot$. Any non-witnessing partial assignment can be extended to such $x$, so we can apply \autoref{lemma: dence-input-subset} (with $X'$ being all inputs). Now we are ready to apply the main theorem with $X_N = X_{N,t(N)}$. Given $p(N) = \poly(\log N)$, as $\mathcal{P}_1$ and $\mathcal{P}_2$ we take sets of partial assignments of size $r(t-1)$ and $k(t-1)$, respectively, with the condition that these partial assignments do not contain $\bot$, and every other value occurs either $0$ or $t-1$ times (see \autoref{diagram:t-pigeon-structure}). Here $r = p(N)$, $k = N - \floor{N^{\eps}}$ for $\eps = \eps(N) = \frac{1}{4t}$.

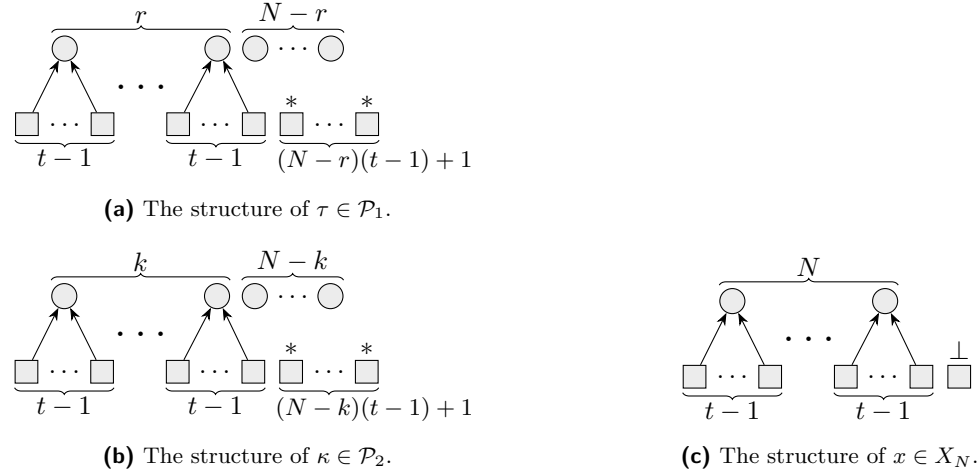
\begin{figure}[ht]
    \centering
    \captionsetup[subfigure]{justification=centering}
    \begin{subfigure}[b]{0.5\textwidth}
    \centering
    \captionsetup[subfigure]{justification=centering}
    \begin{subfigure}[b]{\textwidth}
        \centering
        \begin{tikzpicture}[node distance=1cm]
            \node[square] (left1){};
            \node[square,right of=left1] (right1){};
            \node[] at ($(left1)!0.5!(right1)$) (mid1) {$\ldots$};
            \node[circ,above of=mid1] (up1) {};
            \draw[decoration={brace,mirror,raise=6pt},decorate](left1.west) -- node[below=6pt] {$t-1$} (right1.east);
            \draw[arrow](left1) -- (up1);
            \draw[arrow](right1) -- (up1);

            \node[square, right of=right1] (left2){};
            \node[square,right of=left2] (right2){};
            \node[xshift=0.2mm] at ($(left2)!0.5!(right2)$) (mid2) {$\dots$};
            \node[circ,above of=mid2] (up2) {};
            \draw[decoration={brace,mirror,raise=6pt},decorate](left2.west) -- node[below=6pt] {$t-1$} (right2.east);
            \draw[arrow](left2) -- (up2);
            \draw[arrow](right2) -- (up2);

            \node[] at ($(mid1)!0.5!(up1)$) (center1) {};
            \node[] at ($(mid2)!0.5!(up2)$) (center2) {};
            \node[] at ($(center1)!0.5!(center2)$) {\LARGE $\ldots$};
            \draw[decoration={brace,raise=6pt},decorate](up1.west) -- node[above=6pt] {$r$} (up2.east);

            \node[square,right of=right2,node distance=0.5cm] (ld) {};
            \node[above of=ld,node distance=\shift1] {$*$};
            \node[square,right of=ld] (rd) {};
            \node[above of=rd,node distance=\shift1] {$*$};
            \node at ($(ld)!0.5!(rd)$) {$\ldots$};
            \draw[decoration={brace,mirror,raise=6pt},decorate](ld.west) -- node[below=6pt, xshift=0.6cm] {\small $(N-r)(t-1)+1$} (rd.east);

            \node[circ,right of=up2,node distance=0.5cm] (lu) {};
            \node[circ,right of=lu] (ru) {};
            \node at ($(lu)!0.5!(ru)$) {$\ldots$};
            \draw[decoration={brace,raise=6pt},decorate](lu.west) -- node[above=6pt] {$N-r$} (ru.east);
        \end{tikzpicture}
        \caption{The structure of $\tau \in \mathcal{P}_1$.}
    \end{subfigure}
    
    \vspace{2mm}
    
    \begin{subfigure}[b]{\textwidth}
        \centering
        \begin{tikzpicture}[node distance=1cm]
            \node[square] (left1){};
            \node[square,right of=left1] (right1){};
            \node[] at ($(left1)!0.5!(right1)$) (mid1) {$\ldots$};
            \node[circ,above of=mid1] (up1) {};
            \draw[decoration={brace,mirror,raise=6pt},decorate](left1.west) -- node[below=6pt] {$t-1$} (right1.east);
            \draw[arrow](left1) -- (up1);
            \draw[arrow](right1) -- (up1);

            \node[square, right of=right1] (left2){};
            \node[square,right of=left2] (right2){};
            \node[xshift=0.2mm] at ($(left2)!0.5!(right2)$) (mid2) {$\dots$};
            \node[circ,above of=mid2] (up2) {};
            \draw[decoration={brace,mirror,raise=6pt},decorate](left2.west) -- node[below=6pt] {$t-1$} (right2.east);
            \draw[arrow](left2) -- (up2);
            \draw[arrow](right2) -- (up2);

            \node[] at ($(mid1)!0.5!(up1)$) (center1) {};
            \node[] at ($(mid2)!0.5!(up2)$) (center2) {};
            \node[] at ($(center1)!0.5!(center2)$) {\LARGE $\ldots$};
            \draw[decoration={brace,raise=6pt},decorate](up1.west) -- node[above=6pt] {$k$} (up2.east);

            \node[square,right of=right2,node distance=0.5cm] (ld) {};
            \node[above of=ld,node distance=\shift1] {$*$};
            \node[square,right of=ld] (rd) {};
            \node[above of=rd,node distance=\shift1] {$*$};
            \node at ($(ld)!0.5!(rd)$) {$\ldots$};
            \draw[decoration={brace,mirror,raise=6pt},decorate](ld.west) -- node[below=6pt, xshift=0.6cm] {\small $(N-k)(t-1)+1$} (rd.east);

            \node[circ,right of=up2,node distance=0.5cm] (lu) {};
            \node[circ,right of=lu] (ru) {};
            \node at ($(lu)!0.5!(ru)$) {$\ldots$};
            \draw[decoration={brace,raise=6pt},decorate](lu.west) -- node[above=6pt] {$N-k$} (ru.east);
        \end{tikzpicture}
        \caption{The structure of $\kappa \in \mathcal{P}_2$.}
    \end{subfigure}
    \end{subfigure}
    \hfill
    \begin{subfigure}[b]{0.4\textwidth}
    \centering
    \captionsetup[subfigure]{justification=centering}  
    \begin{subfigure}[b]{\textwidth}
        \centering
        \begin{tikzpicture}[node distance=1cm]
            \node[square] (left1){};
            \node[square,right of=left1] (right1){};
            \node[] at ($(left1)!0.5!(right1)$) (mid1) {$\ldots$};
            \node[circ,above of=mid1] (up1) {};
            \draw[decoration={brace,mirror,raise=6pt},decorate](left1.west) -- node[below=6pt] {$t-1$} (right1.east);
            \draw[arrow](left1) -- (up1);
            \draw[arrow](right1) -- (up1);

            \node[square, right of=right1] (left2){};
            \node[square,right of=left2] (right2){};
            \node[xshift=0.2mm] at ($(left2)!0.5!(right2)$) (mid2) {$\dots$};
            \node[circ,above of=mid2] (up2) {};
            \draw[decoration={brace,mirror,raise=6pt},decorate](left2.west) -- node[below=6pt] {$t-1$} (right2.east);
            \draw[arrow](left2) -- (up2);
            \draw[arrow](right2) -- (up2);

            \node[square,right of=right2,node distance=0.5cm] (lonely) {};
            \node[above of=lonely,node distance=0.35cm] {$\bot$};

            \node[] at ($(mid1)!0.5!(up1)$) (center1) {};
            \node[] at ($(mid2)!0.5!(up2)$) (center2) {};
            \node[] at ($(center1)!0.5!(center2)$) {\LARGE $\ldots$};
            \draw[decoration={brace,raise=6pt},decorate](up1.west) -- node[above=6pt] {$N$} (up2.east);
        \end{tikzpicture}
        \caption{The structure of $x \in X_N$.}
    \end{subfigure}
    \end{subfigure}
    \caption{The structure of elements in $\mathcal{P}_1$, $\mathcal{P}_2$ and $X_N$ for $t$-$\Pigeon_N$ up to a permutation of pigeons and holes. Here squares represent pigeons, circles represent holes, and arrows show how pigeons are mapped to holes.}
    \label{diagram:t-pigeon-structure}
\end{figure}

Let us go through the sub-items of \autoref{theorem:main}:
\begin{itemize}
    \item \ref{item:main-extendability}: For any non-witnessing partial assignment $\rho$ with $|\rho| \leq p(N)$ we can take $A_\rho = \{\rho' \in \mathcal{P}_1 \mid \rho' \sqsupset \rho\}$. It is clear that $x$ extends $\rho$ if and only if it extends some $\rho' \in A_\rho$.
    \item \ref{item:main-isomorphysm}: Note that for every $N$ there is $\hat{N}$ such that $\floor{\hat{N}^{\eps(\hat{N})}} = N$. Moreover:
    \begin{gather*}
        N \geq \hat{N}^{\frac{1}{8t(\hat{N})}} \geq \hat{N}^{\frac{1}{8\sqrt{\log\hat{N}}}} = 2^{\frac{\sqrt{\log\hat{N}}}{8}} \Rightarrow 
        \hat{N} \leq 2^{64(\log N)^2} = 2^{O((\log N)^2)}
    \end{gather*}
    Also, for any $\kappa \in \mathcal{P}_2(\hat{N})$, the search problem $(t(\hat{N})\text{-}\Pigeon_{\hat{N}}, X_{\hat{N},t(\hat{N})})\restrharpoon \kappa$ is equivalent to $(t(\hat{N})\text{-}\Pigeon_N, X_{N,t(\hat{N})})$ up to an alphabet renaming. Now, there is an always successful pseudo-reduction from $(t(N)\text{-}\Pigeon_N, X_{N,t(N)})$ to $(t(\hat{N})\text{-}\Pigeon_N, X_{N,t(\hat{N})})$, achieved by introducing for each of $N$ holes $(t(\hat{N}) - t(N))$ new pigeons, pointing to it. Then, there is an always successful pseudo-reduction from $(t(N)\text{-}\Pigeon_N, X_{N,t(N)})$ to $(t(\hat{N})\text{-}\Pigeon_{\hat{N}}, X_{\hat{N},t(\hat{N})})\restrharpoon \kappa$ with complexity $O(\log \hat{N}) = O((\log N)^2)$. From this follows the required reducibility.
    \item \ref{item:distibutions}: By the symmetry of $\mathcal{P}_2$ and $X_N$.
    \item \ref{item:main-extendability-vs-consistency}: For this one we can use \autoref{lemma:tiling-optimality}. Note that due to symmetry the value $\gamma \coloneqq |\{\kappa \in \mathcal{P}_2 \mid \kappa \agree{X_N} \tau\}|/|\{\kappa \in \mathcal{P}_2 \mid \kappa \sqsupset \tau\}|$ does not depend on choice of $\tau \in \mathcal{P}_1$. Other conditions of lemma also follow from symmetry. Finally, simple combinatorial estimates (see \autoref{appendix:t-pigeon-calculations}) show that $\gamma = (1 + O(N^{-\frac 12})) \Rightarrow \gamma^{-1} = (1 - o(N^{-\frac 13}))$, which concludes the proof.
\end{itemize}

\subsection{PPA}
\begin{definition}
    $\PPA$ is a class of $\TFNP$ problems that are reducible to the search problem $\Lonely = \{\Lonely_N\}_N$. $\Lonely_N$ is defined as follows:\\
    $\Sigma_N = [2N+1]\cup \{\bot\}$, $l_N = 2N+1$, $O_N = [2N+1]$. Correct solution for input $x$: any $i$ such that $x_i \in \{i,\bot\}$ or $x_{x_i} \neq i$.
\end{definition}
First, we restrict our attention to consistent inputs: $X'_N = \{x$ | $(\forall i,j\in[2N+1] : x_i = j \Leftrightarrow x_j = i) \land (\forall i\in [2N+1] : x_i \neq i)\}$. It is easy to see that $\PPA$ is reducible to itself, restricted to such inputs. Now all inputs represent graph, consisting of pairs and lonely nodes. We can further refine $X'$ to $X\subset X'$ such that there is exactly one lonely node. \autoref{lemma: dence-input-subset} shows that $X$ is sufficient for checking reducibility from $\PPA$. 

Now we move to the second part of \autoref{theorem:main}. We need to specify $\mathcal{P}_1$ and $\mathcal{P}_2$, given $p(N) = \poly(\log N)$. We take sets of partial assignments, extendable to inputs in $X_N$, that reveal exactly $r$ and $k$ pairs, respectively (see \autoref{diagram:lonely-structure}). There $r = p(N)$, $k = N - \floor{N^{\eps}}$ for $\eps = \frac 14$. Now we are ready to apply \autoref{theorem:main}. Let us go through the sub-items:
\begin{itemize}
    \item \ref{item:main-extendability}: For any non-witnessing partial assignment $\rho$ with $|\rho| \leq p(N)$ and which is extendable to $x \in X_N$, we can take $A_\rho = \{\rho' \in \mathcal{P}_1 \mid \rho' \sqsupset \rho\}$. It is clear that $x$ extends $\rho$ if and only if it extends some $\rho' \in A_\rho$.
    \item \ref{item:main-isomorphysm}: Note that for every $N$ there is $\hat{N} = O(N^{\frac{1}{\eps}})$ such that $\floor{\hat{N}^{\eps}} = N$. Also, for any $\kappa \in \mathcal{P}_2(\hat{N})$, search problem $(\Lonely_{\hat{N}}, X_{\hat{N}})\restrharpoon \kappa$ is equivalent to $(\Lonely_N, X_N)$ up to renaming the alphabet, which means there is a trivial always successful pseudo-reduction from $(\Lonely_N, X_N)$ to $(\Lonely_{\hat{N}}, X_{\hat{N}})\restrharpoon \kappa$ with complexity $\log \hat{N} = O(\log N)$. From this follows the required reducibility.
    \item \ref{item:distibutions}: By the symmetry of $\mathcal{P}_2$ and $X_N$.
    \item \ref{item:main-extendability-vs-consistency}: Here we can use \autoref{lemma:tiling-optimality}. Note that due to symmetry the value $\gamma \coloneqq |\{\kappa \in \mathcal{P}_2 \mid  \kappa \agree{X_N} \tau\}|/|\{\kappa \in \mathcal{P}_2 \mid \kappa \sqsupset \tau\}|$ does not depend on choice of $\tau \in \mathcal{P}_1$. Other conditions of lemma also follow from symmetry. Finally, simple combinatorial estimates (see \autoref{appendix:lonely-calculations}) show that $\gamma = (1 + O(N^{-\frac 14})) \Rightarrow \gamma^{-1} = (1 - o(N^{-\frac 15}))$, which concludes the proof.
\end{itemize}

\begin{figure}[ht]
    \centering
    \captionsetup[subfigure]{justification=centering}
    \begin{subfigure}{0.3\textwidth}
        \centering
        \begin{tikzpicture}[node distance=1cm]
            \node[circ] (a11){};
            \node[circ,above of=a11] (a21){};
            \node[circ,right of=a11] (a12){};
            \node[circ] (a22) at (a21 -| a12) {};
            \draw[decoration={brace,mirror,raise=6pt},decorate](a11.west) -- node[below=6pt] {$r$} (a12.east);
            \draw[bi-arrow](a11) -- (a21);
            \draw[bi-arrow](a12) -- (a22);
            \node at ($(a11)!0.5!(a12)$) {$\ldots$};
            \node at ($(a21)!0.5!(a22)$) {$\ldots$};

            \node[circ,right of=a12,node distance=0.5cm] (a13) {};
            \node[above of=a13,node distance=\shift1] {$*$};
            \node[circ,right of=a13] (a14) {};
            \node[above of=a14,node distance=\shift1] {$*$};
            \node at ($(a13)!0.5!(a14)$) {$\ldots$};
            \draw[decoration={brace,mirror,raise=6pt},decorate](a13.west) -- node[below=6pt,xshift=2mm] {$2N-2r+1$} (a14.east);
        \end{tikzpicture}
        \caption{The structure of $\tau \in \mathcal{P}_1$.}
    \end{subfigure}
    \hfill
    \begin{subfigure}{0.3\textwidth}
        \centering
        \begin{tikzpicture}[node distance=1cm]
            \node[circ] (a11){};
            \node[circ,above of=a11] (a21){};
            \node[circ,right of=a11] (a12){};
            \node[circ] (a22) at (a21 -| a12) {};
            \draw[decoration={brace,mirror,raise=6pt},decorate](a11.west) -- node[below=6pt] {$k$} (a12.east);
            \draw[bi-arrow](a11) -- (a21);
            \draw[bi-arrow](a12) -- (a22);
            \node at ($(a11)!0.5!(a12)$) {$\ldots$};
            \node at ($(a21)!0.5!(a22)$) {$\ldots$};

            \node[circ,right of=a12,node distance=0.5cm] (a13) {};
            \node[above of=a13,node distance=\shift1] {$*$};
            \node[circ,right of=a13] (a14) {};
            \node[above of=a14,node distance=\shift1] {$*$};
            \node at ($(a13)!0.5!(a14)$) {$\ldots$};
            \draw[decoration={brace,mirror,raise=6pt},decorate](a13.west) -- node[below=6pt,xshift=2mm] {$2N-2k+1$} (a14.east);
        \end{tikzpicture}
        \caption{The structure of $\kappa \in \mathcal{P}_2$.}
    \end{subfigure}
    \hfill
    \begin{subfigure}{0.3\textwidth}
        \centering
        \begin{tikzpicture}[node distance=1cm]
            \node[circ] (a11){};
            \node[circ,above of=a11] (a21){};
            \node[circ,right of=a11] (a12){};
            \node[circ] (a22) at (a21 -| a12) {};
            \draw[decoration={brace,mirror,raise=6pt},decorate](a11.west) -- node[below=6pt] {$N$} (a12.east);
            \draw[bi-arrow](a11) -- (a21);
            \draw[bi-arrow](a12) -- (a22);
            \node at ($(a11)!0.5!(a12)$) {$\ldots$};
            \node at ($(a21)!0.5!(a22)$) {$\ldots$};

            \node[circ,right of=a12,node distance=0.5cm] (a13) {};
            \node[above of=a13,node distance=0.35cm] {$\bot$};
        \end{tikzpicture}
        \caption{The structure of $x \in X_N$.}
    \end{subfigure}
    \caption{The structure of elements in $\mathcal{P}_1$, $\mathcal{P}_2$ and $X_N$ for $\Lonely_N$ up to a permutation of nodes. Here circles represent nodes, and arrows show how they are matched.}
    \label{diagram:lonely-structure}
\end{figure}
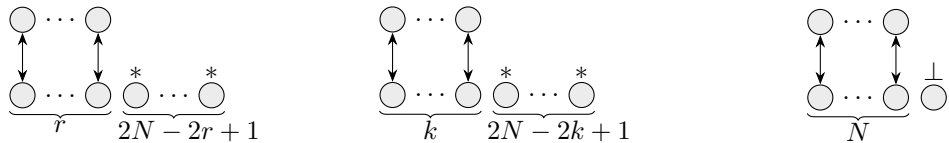

\section{Future Directions}
In addition to all these classes, we also attempted to apply this method to complete problems in classes $\PLS$, $\PPADS$, and $\PWPP$. In most cases, our method fails in combinatorial estimates (condition \ref{item:main-extendability-vs-consistency} of \autoref{theorem:main}). We hope that this method can be generalized to also work for these classes. But it is quite possible that for some of these classes, similarly to $\Lossy$, randomized reducibility from them does not imply deterministic reducibility.

\bibliography{all}
\appendix
\section{Estimate on the number of consistent \texorpdfstring{$\kappa$}{kappa} for \texorpdfstring{$\Pigeon$}{Pigeon}} \label{appendix:pigeon-calculations}
For definitions of $\mathcal{P}_1$, $\mathcal{P}_2$ and $X_N$ refer to \autoref{diagram:pigeon-structure}. There $r=\poly(\log N)$, $k=N-s$, $s = \floor{N^{\eps}}$, $\eps = \frac 14$. For an arbitrary $\tau \in \mathcal{P}_1$, let $a_l = |\{\kappa \in \mathcal{P}_2 \mid \kappa \agree{X_N} \tau$, $\kappa$ and $\tau$ have $r-l$ common non-$*$ positions$\}$. Let us write an exact formula for this value. To choose such $\kappa$, we need to:
\begin{enumerate}
    \item Choose $r-l$ pigeons that are common with $\tau$.
    \item Choose the remaining $k - (r-l)$ pigeons among $N-r$ not used in $\tau$.
    \item Choose the remaining $k - (r-l)$ holes among $N-1-r$ not used in $\tau$.
    \item Choose bijection from $k - (r-l)$ pigeons to $k - (r-l)$ holes.
\end{enumerate}
To put it into combinatorial terms:
\begin{gather*}
    a_l = \binom{r}{r-l}\binom{N-r}{k-(r-l)}\binom{N-1-r}{k-(r-l)}(k-(r-l))! = \\
    \binom{r}{l}\binom{N-r}{s-l}\binom{N-1-r}{s-l-1}(N-s-(r-l))!
\end{gather*}
Now, let us show that $a_l$ gets progressively smaller as $l$ goes from $0$ to $r$:
\begin{gather*}
    \frac{a_{l+1}}{a_l} = O(r)\cdot O\br{\frac{s}{N}}\cdot O\br{\frac{s}{N}}\cdot O(N) = o(N^{-\frac 14}).
\end{gather*}
Note that this estimate is uniform over $l \in \{0,\ldots,r-1\}$. Finally:
\begin{gather*}
    |\{\kappa \in \mathcal{P}_2 \mid \kappa \agree{X_N} \tau\}| = 
    \sum_{l = 0}^{r} a_l = a_0\left(1 + \sum_{l=1}^r \frac{a_l}{a_0}\right) = \\
    |\{\kappa \in \mathcal{P}_2 \mid \kappa \sqsupset \tau\}|\left(1 + o\br{\sum_{l=1}^r N^{-\frac l4}}\right) = \\
    |\{\kappa \in \mathcal{P}_2 \mid \kappa \sqsupset \tau\}|(1 + o(N^{-\frac 14})).
\end{gather*}

\section{Estimate on the number of consistent \texorpdfstring{$\kappa$}{kappa} for \texorpdfstring{$t$-$\Pigeon$}{t-Pigeon}} 
\label{appendix:t-pigeon-calculations}
For definitions of $\mathcal{P}_1$, $\mathcal{P}_2$ and $X_N$ refer to \autoref{diagram:t-pigeon-structure}. There $r = \poly(\log N)$, $k = N-s$, $s \coloneqq \floor{N^{\eps}}$, $\eps = \frac{1}{4t}$, $t \leq \sqrt{\log N}$. For arbitrary $\tau \in \mathcal{P}_1$, let $a_l = |\{\kappa \in \mathcal{P}_2 \mid \kappa \agree{X_N} \tau$, $\kappa$ and $\tau$ have exactly $r-l$ common non-$*$ symbols$\}|$. To choose such $\kappa$, we need to:
\begin{enumerate}
    \item Choose $r-l$ \enquote{$t-1$ pigeons pointing to one hole} groups that are common with $\tau$.
    \item Choose the remaining $k - (r-l)$ holes among $N-r$ not used in $\tau$.
    \item Choose the remaining $(t-1)(k - (r-l))$ pigeons among $(t-1)(N-r)+1$ not used in $\tau$.
    \item Choose how to map $(t-1)(k - (r-l))$ pigeons to $k - (r-l)$ holes so that at each hole point exactly $t-1$ pigeons.
\end{enumerate}
To put it into combinatorial terms:
\begin{gather*}
    a_l = \binom{r}{r-l}\binom{N-r}{k-(r-l)}\binom{(t-1)(N-r) + 1}{(t-1)(k-(r-l))}\frac{((t-1)(k-(r-l)))!}{((t-1)!)^{k-(r-l)}} = \\
    \binom{r}{l}\binom{N-r}{s-l}\binom{(t-1)(N-r) + 1}{(t-1)(s-l) + 1}\frac{((t-1)(N-s-(r-l)))!}{((t-1)!)^{N-s-(r-l)}}.
\end{gather*}
Now, let us show that $a_l$ gets progressively smaller as $l$ goes from $0$ to $r$:
\begin{gather*}
    \frac{a_{l+1}}{a_l} = O(r)\cdot O\br{\frac{s}{N}}\cdot O\br{\frac{((t-1)s)^{t-1}}{((t-1)N)^{t-1}}} \cdot O\br{\frac{((t-1)N)^{t-1}}{(t-1)!}} = \\
    O\br{\frac{rs^t(t-1)^{t-1}}{N(t-1)!}} = 
    O\br{\frac{rs^t}{N}\frac{(t-1)^{t-1}}{\br{\frac{t-1}{e}}^{t-1}}} = O\br{\frac{rs^te^{t-1}}{N}} = 
    o(N^{-\frac 12}).
\end{gather*}
Note that this estimate is uniform over $l \in \{0,\ldots,r-1\}$. Finally:
\begin{gather*}
    |\{\kappa \in \mathcal{P}_2 \mid \kappa \agree{X_N} \tau\}| = \sum_{l=0}^r a_l = 
    a_0\br{1+\sum_{l=1}^r \frac{a_l}{a_0}} = \\
    |\{\kappa \in \mathcal{P}_2 \mid \kappa \sqsupset \tau\}|\br{1 + o\br{\sum_{l=1}^r N^{-\frac l2 }}} = \\
    |\{\kappa \in \mathcal{P}_2 \mid \kappa \sqsupset \tau\}|(1+o(N^{-\frac 12})).
\end{gather*}

\section{Estimate on the number of consistent \texorpdfstring{$\kappa$}{kappa} for \texorpdfstring{$\Lonely$}{Lonely}}
\label{appendix:lonely-calculations}
For definitions of $\mathcal{P}_1$, $\mathcal{P}_2$ and $X_N$ refer to \autoref{diagram:lonely-structure}. There $r = \poly(\log N)$, $k = N-s$, $s \coloneqq \floor{N^{\eps}}$, $\eps = \frac{1}{4}$. For arbitrary $\tau \in \mathcal{P}_1$, let $a_l = |\{\kappa \in \mathcal{P}_2 \mid \kappa \agree{X_N} \tau$, $\kappa$ and $\tau$ have $r-l$ common pairs$\}|$. To choose such $\kappa$, we need to:
\begin{enumerate}
    \item Choose $r-l$ pairs that are common with $\tau$.
    \item Choose the remaining $2(k - (r-l))$ nodes among $2N-1-2r$ not used in $\tau$.
    \item Choose how to divide those $2(k - (r-l))$ nodes into pairs.
\end{enumerate}
To put it into combinatorial terms:
\begin{gather*}
    a_l = \binom{r}{r-l}\binom{2N-1-2r}{2(k-(r-l))}(2(k-(r-l)) - 1)!! = \\
    \binom{r}{l}\binom{2N-1-2r}{2s-2l-1}(2(N-s-r+l) - 1)!!
\end{gather*}
Now, let us show that $a_l$ gets progressively smaller as $l$ goes from $0$ to $r$:
\begin{gather*}
    \frac{a_{l+1}}{a_l} = O(r)\cdot O\br{\frac{s^2}{N^2}}\cdot O(N) = O\br{\frac{rs^2}{N}} = o(N^{-\frac 14}).
\end{gather*}
Note that this estimate is uniform over $l \in \{0,\ldots,r-1\}$. Finally:
\begin{gather*}
    |\{\kappa \in \mathcal{P}_2 \mid \kappa \agree{X_N} \tau\}| = \sum_{l=0}^r a_l = 
    a_0\br{1+\sum_{l=1}^r \frac{a_l}{a_0}} = \\
    |\{\kappa \in \mathcal{P}_2 \mid \kappa \sqsupset \tau\}|\br{1 + o\br{\sum_{l=1}^r N^{-\frac l4 }}} = \\
    |\{\kappa \in \mathcal{P}_2 \mid \kappa \sqsupset \tau\}|(1+o(N^{-\frac 14})).
\end{gather*}
\end{document}